\documentclass{article}

\usepackage[active]{srcltx}

\usepackage{color}

\usepackage[active]{srcltx}
\usepackage{a4wide}

\usepackage{amsmath}
\usepackage{amsfonts}
\usepackage{amssymb}
\usepackage{amscd}

\newcounter{theorem}
\renewcommand{\thetheorem}{\arabic{section}.\arabic{theorem}}

\newenvironment{thm}[1]{\par
\begin{sloppypar}\refstepcounter{theorem}%
\noindent{\bf #1 \thetheorem.}\it{}}{\end{sloppypar}}

\newenvironment{theorem}{\begin{thm}{Theorem}}{\end{thm}}
\newenvironment{proposition}{\begin{thm}{Proposition}}{\end{thm}}
\newenvironment{corollary}{\begin{thm}{Corollary}}{\end{thm}}
\newenvironment{lemma}{\begin{thm}{Lemma}}{\end{thm}}

\newenvironment{defi}[1]{\par
\begin{sloppypar}\refstepcounter{theorem}%
\noindent{\bf #1 \thetheorem.}\rm{}}{\end{sloppypar}}

\newenvironment{remark}{\begin{defi}{Remark}}{\end{defi}}
\newenvironment{hypothesis}{\begin{defi}{Hypothesis}}{\end{defi}}

\newcommand{\eh}{\hfill}\newlength{\sperr}
\newcommand{\R}{\mathbb{R}}
\newcommand{\x}{{\bf x}}
\newenvironment{proof}{{\settowidth{\sperr}{\rm Proof}
\par\addvspace{0.3cm}\noindent\parbox[t]{1.3\sperr}{\rm P\eh r\eh o\eh o\eh 
f\eh.}}}{\nopagebreak\mbox{}\hfill $\blacksquare $\par\addvspace{0.25cm}}

\def\L{\boldsymbol{\mathcal{L}}}
\def\C{\boldsymbol{\mathcal{C}}}
\def\D{\boldsymbol{\mathcal{D}}}
\def\I{\boldsymbol{\mathcal{I}}}
\def\S{\boldsymbol{\Sigma}}

\def\Im{{{\rm Im}\,}}
\def\vsth{\vskip2mm}

\begin{document}

\noindent 
\begin{center}
\textbf{\large Adiabatic non-equilibrium steady states in the
  partition free approach}
\end{center}

\begin{center}
June 22, 2010
\end{center}

\vspace{0.5cm}

\noindent 

\begin{center}
\small{ 
Horia D. Cornean\footnote{Department of Mathematical Sciences, 
    Aalborg
    University, Fredrik Bajers Vej 7G, 9220 Aalborg, Denmark; e-mail:
    cornean@math.aau.dk},
\framebox{Pierre Duclos}, 
Radu Puric${\rm e}^{1,}$\footnote{
Inst. of Math. ``Simion Stoilow'' of
the Romanian Academy, P. O. Box 1-764, RO-014700 Bu\-cha\-rest,
Romania;
e-mail: Radu.Purice@imar.ro}}
     
\end{center}

\vspace{0.5cm}

\noindent

\begin{abstract}
Consider a small
sample coupled to a finite number of
leads, and assume that the total (continuous) system is at thermal equilibrium in the remote past. We 
construct a non-equilibrium steady state (NESS) by adiabatically turning on an electrical bias
between the leads. The main mathematical challenge is to show that certain 
adiabatic wave operators exist, and to identify their strong limit when the
adiabatic parameter tends to zero. Our NESS is different from, though
closely related with the NESS provided by the Jak{\v
  s}i{\'c}-Pillet-Ruelle approach. Thus we partly settle a question asked by Caroli {\it et al} in 1971 regarding the (non)equivalence between the partitioned and partition-free approaches.
\end{abstract}

\tableofcontents

\newpage

\section{Introduction}\label{Intro}

\setcounter{equation}{0}

\subsection{Generalities}

This paper deals with the rigorous construction of adiabatic non-equilibrium steady states
for mesoscopic systems which initially are fully coupled (or 'partition free')
and at thermal equilibrium \cite{Cini, CGZ}. The initial equilibrium
state is broken
down by slowly turning on an electrical bias between leads
(i.e. inserting a d.c. battery), which in a certain way can be seen as
slowly changing the chemical potentials of the leads coupled with the
small sample. 

In contrast with the above described partition-free setting, 
the 'partitioned procedure' is the one in which 
one starts with several decoupled reservoirs, each of them being at 
different equilibrium states. Let us assume for simplicity
that they are in grand canonical Gibbs states having the same temperature but
different chemical potentials. Then at $t=0$ they are
suddenly joined together with a sample, and the newly composed system 
is allowed to freely evolve until it reaches a steady state at
$t=\infty$. From a mathematical point of view this approach is by now
very well understood, see for example \cite{AJPP, JP1, Ruelle, N, brarob} and
references therein.  One can 
allow the carriers to interact in the sample \cite{JOP}, and the
theory still works. Note that even if we
choose to turn on the coupling
between the reservoirs in a time dependent way, 
the result will be the same \cite{CNZ}. 

One can ask which approach is more physical; here is a quote from a
paper by Caroli {\it et al} \cite{Caroli} from 1971 -maybe the
first very influential paper on the subject- who came with the following
observation about the partitioned procedure: 
{\it One might raise a major objection to the above procedure; 
it amounts to establishing
first the dc bias, and only later the coupling between the 
barrier and the electrode. Physically,
it is the reverse that is true; the 
transfer matrix elements are always there, 
and the dc bias
is established afterwards; it is not obvious that the 
corresponding limits can be interchanged.}

The major achievement of our current paper is that we can now
construct an adiabatic NESS in the partition free setting; let us explain how.  
The leads are already coupled with the sample, and at $t=-\infty$ the
full system is in a Gibbs equilibrium state at a given temperature and 
chemical potential. Then we adiabatically turn on a potential bias
$V\chi(\eta t)$ between the leads, modeling
in this way a gradual appearance of a difference in the chemical
potentials (here $\chi(-\infty)=0$, 
$\chi(0)=1$ and $\eta>0$ is the adiabatic parameter). 
The final bias $V$ does not need to be small; 
our results are beyond the linear response theory. 
The statistical density matrix $\rho_\eta(t)$ is found as the solution of
a quantum Liouville equation, with the initial condition at 
$t=-\infty$ given by the 
global Gibbs state. 

In Theorem \ref{main} we show the existence and
compute the strong limit $\rho_{ad}:=\lim_{\eta\searrow 0}\rho_\eta(t)$. The limit is 
$t$ independent, and contains - as in the partitioned procedure- 
two contributions: one from the discrete, and one from the continuous
subspaces. Note that we do not have to take the Ces{\`a}ro limit in order
to insure convergence for the discrete part. The adiabatic limit takes
care of the oscillations. The price we pay is that we need to demand
that the point spectrum of certain Hamiltonians only consists from
finitely many discrete eigenvalues. Most probably this condition is
too strong, and getting rid of it remains an interesting open problem. 

Even though the stationary density matrix of 
the partitioned procedure has a similar structure, it is {\bf
  different} from the one we construct here. A careful comparison will be given elsewhere.  

A future problem is to investigate the charge current and
establish Landauer-B\"uttiker type formulas \cite{But, BPT, Avron1,
  Avron2, CGZ} in the partition free setting with a continuous model, 
and without the linear response approximation. In fact this
was the starting point of a number of remarkable physical papers,
see for example \cite{FL}, \cite{LA}, \cite{BS}. A first
mathematically sound derivation of the L-B formula on a discrete model
and under the linear response approximation was obtained in
\cite{CJM1} and further investigated in
\cite{CJM2}. In \cite{CDNP-1} we significantly improved the method of proof of
\cite{CJM1}, which also allowed us to extend the results to the continuous
case. 

Another challenging open problem is to extend the formalism in order
to accommodate transient regimes (see \cite{MGM, MGM2, MSS, CGZ} and references
therein), and 
locally interacting fermions \cite{Stefa1, Stefa2}. 

Finally, we want to stress that some of the technical conditions which
we impose for our model (like smoothness of boundaries and potentials,
working with only two parallel leads) can be relaxed. We chose though to work
under stronger conditions in order to give shorter proofs for 
certain spectral and 
asymptotic completeness results, thus making the paper rather self-consistent.
In this way, the number of generic assumptions is kept to a minimum.

\subsection{The model}\label{System}

Take two identical semi-infinite cylinders and couple them 
smoothly through a finite
domain. The cylinders will model the leads, while the connecting domain 
will represent the region where the interesting physics takes 
place. The total configuration space $\boldsymbol{\mathcal{L}}$ is a subset
of $\R^{d+1}$ with $d\geq 0$. In order to simplify presentation, we will
assume that $\boldsymbol{\mathcal{L}}$ is cylinder-like, which means
that for each value of the longitudinal coordinate 
$x_{||}\in\R$ the transverse coordinate $\x_\perp$ belongs to a
bounded cross-section $\mathfrak{D}(x_{||})\subset\mathbb{R}^d$. Again
for the sake of simplicity, we assume that the boundary
\begin{align}\label{finalprima3}
\S:=\partial \boldsymbol{\mathcal{L}}
\end{align}
defines a regular $C^\infty$-surface embedded in $\mathbb{R}^{d+1}$.

Let us start with the description of the configuration space 
associated to one of our $d+1$ dimensional leads,
namely the left one. Let $\tilde{a}>0$. We let
$\widetilde{\I}_-:=(-\infty,-\tilde{a})$ model its
longitudinal dimension. Then we assume that:
$$\boldsymbol{\mathcal{L}}\cap \{\widetilde{\I}_-\times \R^d\}=
:\widetilde{\I}_-\times \mathfrak{D},$$ 
where the transverse section 
$\mathfrak{D}\subset\mathbb{R}^d$ is supposed to be 
a bounded and simply connected open
set with a regular $C^\infty$-boundary $\partial\mathfrak{D}$. Thus the
configuration space of the left cylinder is modeled in a
natural way by the set $\tilde{\I}_-\times\mathfrak{D}$. Similarly, if
$\widetilde{\I}_+:=(\tilde{a},\infty)$,  the configuration space of the right cylinder
is modeled by $\widetilde{\I}_+\times\mathfrak{D}$. 

Now define: 
\begin{align}\label{finalprima1}
\widetilde{\C}:=\boldsymbol{\mathcal{L}}\cap \{[-\tilde{a},\tilde{a}]\times \R^d\}.
\end{align}
Thus the small sample is contained by a bounded and 
simply connected set $\widetilde{\C}\subset\mathbb{R}^{d+1}$
which is smoothly glued to the two leads. 
With these notations, the one particle configuration space can be
decomposed as: 
\begin{align}\label{finalprima2}
\boldsymbol{\mathcal{L}}=\big(
\widetilde{\I}_-\times\mathfrak{D}\big)
\cup
\widetilde{\C}
\cup
\big(
\widetilde{\I}_+\times\mathfrak{D}\big).
\end{align}
When we refer to the "coupled system", we mean that there are no
internal walls between
the sample and leads. A particle will be free to flow inside the
system, and to pass from one
lead to another via the sample. But it is not allowed to get out of 
$\boldsymbol{\mathcal{L}}$.  

Now let us introduce the one particle Hamiltonian of the coupled
system. In the
sample $\widetilde{\C}$ we assume the existence of a potential $w\in
C^\infty_0(\widetilde{\C})$, which will be considered positive 
without loss of generality. The kinetic energy of a particle living in
$\L$ will be modeled by the Laplace operator
$-\Delta_D$ with Dirichlet boundary conditions on 
$\partial\boldsymbol{ \mathcal{L}}$ and having the
domain $\mathbb{H}_D(\L):=H^1_0(\L)\cap H^2(\L)$.
Thus the one-particle Hamiltonian is of the form:
\begin{equation}\label{1-p-hamilt}
H:=-\Delta_D+w,
\end{equation}
with the same domain. 

Regarding the spectral properties of $H$, we will prove in Lemma
\ref{LAP-K1} that its singular continuous
spectrum is absent. We will assume that the pure point
spectrum consists of discrete and finitely many eigenvalues: 
\begin{equation}\label{Hyp-2}
\sigma_{pp}(H)=\sigma_{disc}(H), 
\quad\#\sigma_{pp}(H)<\infty. 
\end{equation}
\begin{remark}\label{absenceOfEmbeddedEigenvalues}%
This assumption means in particular that we do not allow embedded 
eigenvalues in the continuous spectrum. To our knowledge, sufficient 
conditions to guarantee  this property are not known in dimension $d+1\ge2$.
\end{remark}%

\vsth
Let $\mathcal{H}:=L^2(\L)$, and let $a>\tilde{a}$. Define 
\begin{align}\label{1-2010}
\L_-:=\L\cap\{(-\infty,-a)\times \D\},\quad \L_+:=\L\cap 
\{(a,\infty)\times \D\},\quad \C:=\boldsymbol{\mathcal{L}}\cap
\{(-a,a)\times \R^d\}.
\end{align}
We introduce three orthogonal projections:
\begin{align}\label{H-dec}
&\Pi_-:\mathcal{H}\rightarrow\mathcal{H}_-:=
L^2(\L_-),\quad 
\Pi_+:\mathcal{H}\rightarrow\mathcal{H}_+:=L^2(\L_+),
\nonumber\\
&\Pi_0:\mathcal{H}\rightarrow\mathcal{H}_0:=L^2(\C).
\end{align}
Note that $\widetilde{\C}$ is completely included in the open set 
$\C$, and $\L_\pm$ are "shorter" than the
corresponding leads. 

\subsection{The state and the Liouville equation}%

We only work at the level of density matrices. 
In the remote past $t\rightarrow-\infty$ the electron gas is at equilibrium
at a temperature $T>0$ and a chemical potential $\mu$, 
moving in all the volume $\L$. The gas is described by a quasi-free
state, having as two-point function the usual Fermi-Dirac equilibrium
density matrix operator:
\begin{equation}\label{FD-d}
\rho_{eq}(H):=\frac{1}{1+e^{(H-\mu)/kT}}.
\end{equation}
The system is driven out of equilibrium by {\it slowly} turning on an 
electric bias
\begin{equation}\label{defPotBias}
V=v_-\Pi_-+v_+\Pi_+,
\end{equation}
where $v_\pm$ are real constants. We want to introduce the bias 
adiabatically with
an adiabatic parameter $\eta >0$, as
a time-dependent potential $V_\eta(t):=\chi(\eta t)V$. One should have
in mind $\chi(t)=e^{t}$, but only a few abstract
properties of this function are really
needed, namely:
\begin{align}\label{cdh1}
& 0<\chi(t)< 1\quad {\rm and}\quad \chi'(t)>0 \quad {\rm if}\quad t<0;
\qquad \chi(0)=1; \\
&\chi, |\chi''|\in L^1(\mathbb{R}_-).  \nonumber
\end{align}
We will also need to consider the 'bias' with a fixed coupling
constant $\kappa\in[0,1]$. We introduce a family of operators:
\marginpar{$H(\kappa)$}
\begin{equation}\label{fix-bias-hamilt}
K(\kappa):=H+\kappa V.
\end{equation}
$E_{pp}(A)$ and $E_{ac}(A)$ will denote 
respectively the projector on the pure point and absolutely continuous 
spectral subspace of the self-adjoint operator $A$.  In Lemma
\ref{LAP-K1} we will prove that the singular continuous spectrum of 
$K(\kappa)$ is empty. We now make the
following assumptions concerning the point spectrum:

\vsth
\begin{hypothesis}\label{Hyp-3}%
\begin{enumerate}
\item $\forall\kappa\in[0,1]$ the Hamiltonian $K(\kappa)$ 
has no eigenvalues embedded in the continuous spectrum; \item $\dim
   E_{pp}(K({\kappa}))=N<\infty,\;\forall\kappa\in[0,1]$, 
$\sigma_{pp}(K({\kappa}))=\{\varepsilon_j(\kappa)\}_{j=1}^N$;

\item
  $\underset{\kappa\in[0,1]}{\min}\left\{\text{dist}\left(\sigma_{pp}(K(\kappa)), 
\sigma_{ac}(K(\kappa)) \right)\right\}\geq d>0$.
\end{enumerate}
\end{hypothesis}%
In order to simplify the presentation, we will only work with $N=2$
and adopt an extra assumption: 
\vsth
\begin{hypothesis}\label{Hyp-33}

The eigenvalues $\{\varepsilon_j(\kappa)\}_{j\in\{1,2\}}$ (which are 
real analytic functions of $\kappa\in [0,1]$) can cross at most at one
point $\kappa_0\in(0,1)$. This $\kappa_0$ corresponds to some
unique $t_0<0$ where $\chi(t_0)=\kappa_0$ and  $\chi'(t_0)>0$. 
\end{hypothesis}%

\vspace{0.5cm}

The time dependent Hamiltonian will be
\begin{equation}\label{t-dep-hamilt}
K({\chi(\eta t)}):=H+\chi(\eta t)V,
\end{equation}
having the constant domain equal to the domain of $H$,
i.e. $\mathbb{H}_D(\L)$. The evolution defined by the time-dependent
Hamiltonian $K({\chi(\eta t)})$ is described by a unitary propagator
$W_{\eta}(t)$, solution of the following Cauchy problem:
\begin{equation}\label{W}
\left\{
 \begin{array}{l}
  i\partial_tW_{\eta}(t)=K({\chi(\eta t)})W_{\eta}(t)\\
  W_{\eta}(0)=1,
 \end{array}
\right.
\end{equation} 
for $t\in\mathbb{R}$. For any $\eta>0$, the family
$\{K({\chi(\eta t)})\}_{t\in\mathbb{R}}$ consists of self-adjoint
operators in $\mathcal{H}$ having a common domain equal to
$\mathbb{H}_D(\L)$ and strongly differentiable with respect to $t\in \mathbb{R}$
with a bounded self-adjoint norm derivative $
\partial_tK(\chi(\eta t))\,=\,\eta\,\chi'(\eta t)V$. 

Now using well known results quoted in \cite[Th. X.70]{RS2} 
we easily obtain that the problem 
\eqref{W} has a unique solution which is unitary and leaves
 the domain $\mathbb{H}_D(\L)$ invariant for any
 $t\in\mathbb{R}$. Moreover, its adjoint satisfies the equation: 
\begin{equation}\label{finalprima4}
i\partial_tW_{\eta}^*(t)=-W_{\eta}^*(t)K(\chi_{\eta}(t)).
\end{equation}

The object we are interested in is the time evolved density matrix $\rho_\eta(t)$
which must be a solution of the
Liouville equation, starting from the initial value $\rho_{eq}(H)$ at $t\rightarrow-\infty$:
\begin{equation}\label{Liouv}
i\partial_t \rho_\eta(t)=[K(\chi(\eta t)),\rho_\eta(t)],
\quad \underset{t\to -\infty} {n-\lim}\rho_\eta(t)=\rho_{eq}(H).
\end{equation}
In the remaining part of our paper we will show that the unique
solution $\rho_\eta(t)$ of the Liouville equation has a strong limit
when $\eta\searrow 0$, and compute it. In particular, we will see that
the adiabatic limit is $t$ independent.

\subsection{The main result}%

In order to formulate our main result we need to
define some new objects. First, we  introduce the decoupled
Hamiltonian obtained from $H$
by introducing Dirichlet walls where the bias is discontinuous
($x_{||}=\pm a$). Remember that the decomposition \eqref{1-2010}
depends on $a$, and the walls are inside the leads. 
Let
$\overset{\circ}{\Delta}_D$ be the self-adjoint Laplace operator
defined in $L^2(\L)$ with Dirichlet
conditions on $\partial\L_\pm\cup \partial{\C}$; we have 
$
\overset{\circ}{\Delta}_D=\overset{\circ}{\Delta}_{D,-}\oplus\overset{\circ}{\Delta}_{D,0}\oplus\overset{\circ}{\Delta}_{D,+}$,
where their domains are denoted as follows: 
\begin{align}\label{decoupl-domains}
\mathbb{H}_D(\L_\pm)&:=H^1_0(\L_\pm)\cap H^2(\L_\pm),\quad 
\mathbb{H}_D(\C):=H^1_0(\C)\cap H^2(\C),\nonumber \\
\overset{\circ}{\mathbb{H}}_D(\L)&:=\mathbb{H}_D(\L_-)\oplus\mathbb{H}_D(\C)\oplus\mathbb{H}_D(\L_+).
\end{align}
Let us note that due to the cylindrical symmetry of the
regions $\L_{\pm}$ where the bias is piecewise constant,  we can write 
\begin{equation}\label{sect-Lapl}
\overset{\circ}{\Delta}_{D,\pm}=\mathfrak{l}_\pm\otimes1+1\otimes\mathfrak{L}_{\D}
\end{equation}
with $\mathfrak{L}_{\D}$ the Laplacean on the bounded domain 
$\D\subset\mathbb{R}^d$ with Dirichlet conditions on the boundary
$\partial\D$, and $\mathfrak{l}_\pm$ the operator of second derivative
on $\I_{\pm}$ with Dirichlet condition at $\pm a$. The decoupled one particle Hamiltonian will be: 
\begin{equation}\label{dec-hamilt}
\overset{\circ}{H}:=-\overset{\circ}{\Delta}_D +w,
\end{equation}
which is self-adjoint on the domain
$\overset{\circ}{\mathbb{H}}_D(\L)$, having Dirichlet conditions on 
$\partial\L_\pm\cup\partial\C$.  As in the coupled case, 
we need to consider the bias 
with a fixed coupling constant $\kappa\in[0,1]$ and define $
\overset{\circ}{K}(\kappa):=\overset{\circ}{H}+\kappa V$. 
In order to formulate our main theorem we need the following lemma:
\begin{lemma}\label{lemamain}%
The stationary wave operator $\Xi_0$ associated to the pair 
$\{\overset{\circ}{K}(1),K(1)\}$:
$$
\Xi_0:=\underset{s\searrow-\infty}{s-\lim}e^{isK(1)}e^{-is\overset{\circ}{K}(1)}
E_{ac}(\overset{\circ}{K}(1)),
$$
exists and is a
unitary operator from $E_{ac}(\overset{\circ}{K}(1))\mathcal{H}$ to 
$E_{ac}(K(1))\mathcal{H}$. Moreover, the singular continuous spectrum
of $K(\kappa)$ is empty for all $\kappa\in [0,1]$. 
\end{lemma}
\vspace{0.5cm}

And here is the main result:
\begin{theorem}\label{main}%
The adiabatic limit of the density matrix exists
  in the strong operator topology on $\mathbb{B}(\mathcal{H})$, is
  independent of $t$ and given by:
\begin{equation}\label{finalprima6}
\rho_{ad}:=\underset{\eta\searrow 0}{s-\lim}\;\rho_{\eta}(t)\,=\Xi_0
\rho_{eq}(\overset{\circ}{H})\Xi_0^*+\sum_{j=1}^N\rho_{eq}(\varepsilon_j(0))E_j(K(1)),
\end{equation}
where $\{\varepsilon_j(0)\}_{j=1}^N$ are the eigenvalues of $H=K(0)$ in
ascending order, while $\{E_j(K(1))\}_{j=1}^N$ are the eigenprojections of $H+V=K(1)$
obtained by analytically continuing $\{E_j(K(\kappa))\}_{j=1}^N$ from
$\kappa=0$ to $\kappa=1$.
\end{theorem}%

\vspace{0.5cm}

\begin{remark}. Even though Lemma \ref{lemamain} is not
surprising, its proof is not straightforward. 

We also note that the 
adiabatic limit $\rho_{ad}$ commutes with $K(1)=H+V$, but it is not a function of
$K(1)$.  Even though $\rho_\eta(t)$ is a solution of a Liouville
equation involving operators with no internal Dirichlet boundaries at 
$\pm a$, the
limit $\rho_{ad}$ is expressed with the help of a  
comparison operator $\overset{\circ}{H}+V$, depending on $a$, and 
which appears naturally in the proof. 

We will assume $N=2$, but the result holds true for any finite $N$. An
interesting open problem is to study the case $N=\infty$ and when the
eigenvalues can enter the continuous spectrum while $\kappa$ grows 
from $0$ to $1$. Another interesting situation is the one in which we
have a degeneracy at $\kappa=0$;  this situation is related to the Gell-Mann
and Low theorem for degenerate unperturbed states \cite{Panati}. 
\end{remark}
\subsection{A useful expression of the density matrix}%

Before actually starting the study of the adiabatic limit, let us very
quickly show that \eqref{Liouv} has a solution, which can be put into 
a form which is 
particularly convenient for taking the adiabatic limit.

Define the unitary adiabatic wave operators
\begin{equation}\label{finaladoua1}
 \omega_\eta:=\underset{t\to -\infty}{n-\lim}\; W_\eta^*(t)e^{-it H},\quad
 \omega_\eta^*=\underset{t\to -\infty}{n-\lim}\; e^{it H}W_\eta(t).
\end{equation}
They converge in norm due to the following estimate ($s<t$):
\begin{equation}\label{Cauchy-1}
\left\Vert W_\eta^*(t)e^{-it H}\,-\,W_\eta^*(s)e^{-is H}\right\Vert\,\leq\,\int_s^t\left\Vert \frac{d}{d\tau}\left \{W_\eta^*(\tau)e^{-i\tau H}\right\}\right\Vert d\tau\leq ||V||\int_s^t\chi(\eta \tau),
\end{equation}
where we use that $\chi\in L^1(\mathbb{R}_-)$. Then by direct computation we can prove that the operator 
\begin{equation}
\rho_\eta(t):=W_\eta(t)\omega_\eta\rho_{eq}(H)\omega_\eta^*W_\eta^*(t)
\end{equation}
solves the Liouville equation. It also obeys the initial condition
because we can write:
\begin{align}
0&=\lim_{t\to -\infty}\left \Vert \rho_\eta(t)-e^{-itH}
\left \{e^{itH}W_\eta(t)\right\}
\omega_\eta\rho_{eq}(H)\omega_\eta^*\left\{W_\eta^*(t)
e^{-itH}\right\}e^{itH}\right \Vert\nonumber \\
&=\lim_{t\to -\infty}\left \Vert
  \rho_\eta(t)-e^{-itH}\rho_{eq}(H)e^{itH}\right 
\Vert=\lim_{t\to -\infty}\left \Vert
  \rho_\eta(t)-\rho_{eq}(H)\right 
\Vert.
\end{align}
The above solution can be rewritten as:
\begin{equation}\label{hc1prima1}
 \rho_\eta(t)=W_\eta(t)\rho_\eta(0)W_\eta^*(t),
\end{equation}
where 
\begin{equation}\label{hc1prima2}
\rho_\eta(0)=\omega_\eta\rho_{eq}(H)\omega_\eta^*.
\end{equation}

Now let us show that it is enough to prove \eqref{finalprima6} for
$t=0$. Indeed, once this formula is proved for $t=0$ it 
shows that the strong limit of $\rho_\eta(0)$ when
$\eta\searrow 0$ is commuting with $K(1)=H+V$. 
It is elementary to check that $W_\eta(t)$ and $W_\eta^*(t)$ converge
in norm to $e^{-itK(1)}$ and respectively $e^{itK(1)}$ when
$\eta\searrow 0$ (with  $t$ fixed). Since $e^{\pm itK(1)}$ commutes
with $\underset{\eta\searrow 0}{s-\lim}\;\rho_\eta(0)$ it follows that
the adiabatic strong limit of
$\rho_\eta(t)$ must also exist and equal the r.h.s of (\ref{finalprima6}).

Moreover, due to the fact that the limits in \eqref{finaladoua1} are in operator
norm, it is easy to show that we have the identity:
\begin{equation}\label{state-t-2}
\rho_\eta(0)=\underset{s\to -\infty}{n-\lim} \; W_\eta^*(s)\rho_{eq}(H)W_\eta(s).
\end{equation}
It is important to note that the above norm limit is not uniform in
$\eta$, and this is the reason why the adiabatic limit is not
straightforward. Formula \eqref{state-t-2} will be the 
starting point in what follows, and we will be interested in computing 
the double limit:
\begin{equation}\label{state-t-22}
\rho_{ad}=\underset{\eta\searrow 0}{s-\lim}\rho_\eta(0)=\underset{\eta\searrow 0}{s-\lim}\left 
\{\underset{s\to -\infty}{n-\lim} \; W_\eta^*(s)\rho_{eq}(H)W_\eta(s)\right \}.
\end{equation}

\section{A road map of the proof of the adiabatic limit}\label{Ad-Lim}

\setcounter{equation}{0}

 Since our proof of the adiabatic limit is quite long, in this section
 we will give a list of technical results leading to it and postpone their
proofs for the next sections. 

The two terms of
\eqref{finalprima6} are coming from different spectral subspaces of
$H+V$: the first
one from the absolutely continuous spectrum, and the second one from 
the discrete spectrum.

In Lemma \ref{LAP-K1} we will prove the absence of singular continuous
spectrum for $K({\kappa})$, thus we can consider the orthogonal decompositions
\begin{equation}\label{finaladoua6}
\mathcal{H}\,=\,E_{ac}(\kappa)\mathcal{H}\oplus
\left\{\underset{1\leq j\leq N}{\oplus}E_j(\kappa)\mathcal{H}\right\},
\end{equation}
where $E_{ac}(\kappa):=E_{ac}(K(\kappa))$ and
$E_j(\kappa):=E_j(K({\kappa}))$. Let us remark here the important fact
that due to the Rellich Theorem (Theorem II.61 in \cite{K}) 
we can choose the eigenprojections of $K(\kappa)$ to be real analytic
functions of $\kappa$ on the interval $[0,1]$. Then we can write
$$
W_\eta^*(s)\rho_{eq}(H)W_\eta(s)=W_{\eta}^*(s)\rho_{eq}(H)E_{ac}(0)W_{\eta}(s)+
\left\{\underset{1\leq j\leq
    N}{\sum}\rho_{eq}(\varepsilon_j(0))\;W_{\eta}^*(s)
E_j(0)W_{\eta}(s)\right\}.
$$
We will separately take the double limit as in 
\eqref{state-t-22} for both above terms. 

\subsection{The contribution of the discrete spectrum}

Let us start our analysis with the pure-point part and compute
$$
\underset{\eta\searrow0}{s-\lim}\left[\underset{s\searrow-\infty}{n-\lim}
W_{\eta}(s)^*E_j(0)W_{\eta}(s)\right].
$$

As $V$ is a bounded analytic perturbation of $H$, the map $[0,1]\ni
\kappa\mapsto E_j(\kappa)$ is - in particular - 
Lipschitz continuous in the uniform topology. Thus there exists a constant $C>0$ such that:
\begin{equation}\label{E-1}
\left\|W_{\eta}^*(s)E_j(0)W_{\eta}(s)-W_{\eta}^*(s)E_j(\chi(\eta
  s))W_{\eta}(s)\right\|\leq C\chi(\eta s),\quad s\leq 0.
\end{equation}
Thus we can replace $E_j(0)$ with the analytically continued
projection $E_j(\chi(\eta
  s))$ and the limit does not change. 
We will prove the following result (a weaker version of the gap-less
adiabatic theorem, see \cite{teufel} and references therein):
\vskip2mm
\begin{proposition}\label{P-Discr} 
Under our Hypothesis \ref{Hyp-3} the following limit exists in the
uniform topology and we have the equality:
$$
\underset{\eta\searrow0}{n-\lim}\left[\underset{s\searrow-\infty}{n-\lim}
W_{\eta}^*(s)E_j(\chi(\eta s))W_{\eta}(s)\right]=E_j(1),
$$
\end{proposition}%
\noindent which combined with \eqref{E-1} immediately gives:
\vskip2mm
\begin{corollary}\label{Cor-P-discr}%
 \begin{align}\label{hc2prima1}
\underset{\eta\searrow0}{n-\lim}\left[\underset{s\searrow-\infty}
{n-\lim}\; W_{\eta}^*(s)E_j(0)W_{\eta}(s)\right]=E_j(1)
\end{align}
and  
\begin{align}\label{hc2prima2}
&\underset{\eta\searrow0}{n-\lim}
\left[\underset{s\searrow-\infty}{n-\lim}\;
e^{isH}E_{ac}(H)W_{\eta}(s)E_{pp}(K(1))\right]=0.
\end{align}
\end{corollary}%
While \eqref{hc2prima1} concludes the proof of the adiabatic limit for
the discrete part of the spectrum (even in the uniform topology), the
limit in \eqref{hc2prima2} is a technical result which will play a role
in the contribution of the continuous spectrum.

\subsection{The contribution of the continuous spectrum}

We will now focus our attention on the term coming from the absolutely continuous part of
the spectrum:
\begin{equation}\label{E-2-0}
\underset{\eta\searrow 0}{s-\lim}\left[\underset{s\searrow-\infty}
{s-\lim}\;W_{\eta}^*(s)\rho_{eq}(H)E_{ac}(H)W_{\eta}(s)\right].
\end{equation}

Due to \eqref{hc2prima2} we may conclude that
\begin{align}\label{hc3prima1}
&\underset{\eta\searrow0}{s-\lim}\left[\underset{s\searrow-\infty}{s-\lim}W_{\eta}^*(s)
\rho_{eq}(H)E_{ac}(H)W_{\eta}(s)\right]\nonumber \\
&=\underset{\eta\searrow0}{s-\lim}\left[\underset{s\searrow-\infty}{s-\lim}W_{\eta}^*(s)
E_{ac}(H)e^{-isH}
\rho_{eq}(H)e^{isH}E_{ac}(H)W_{\eta}(s)\right]\nonumber \\
&=\underset{\eta\searrow 0}{s-\lim}\left[\underset{s\searrow-\infty}
{s-\lim}E_{\rm ac}(K(1))W_{\eta}^*(s)\rho_{eq}(H)E_{ac}(H)W_{\eta}(s)E_{ac}(K(1))\right],
\end{align}
provided that the last double strong limit exists. Note that all errors go to
zero in the uniform norm. 

The next step in the proof is to replace $\rho_{eq}(H)$  with 
$\rho_{eq}(\overset{\circ}{H})E_{\rm ac}(\overset{\circ}{H})$ in
\eqref{hc3prima1}. In order to show that we can do that replacement, 
let us write the identity:
\begin{align}
&\{\rho_{eq}(\overset{\circ}{H})E_{\rm
  ac}(\overset{\circ}{H})-\rho_{eq}(H)\}
E_{\rm ac}(H)W_\eta(s) \\
&=-\rho_{eq}(\overset{\circ}{H})
E_{\rm pp}(\overset{\circ}{H})e^{-is H}E_{\rm ac}(H)
\left\{e^{is H}W_\eta(s)\right \}+\{\rho_{eq}(\overset{\circ}{H})
-\rho_{eq}(H)\}e^{-is H}E_{\rm ac}(H)\left\{e^{is H}W_\eta(s)\right \}.\nonumber
\end{align}
When $s\to -\infty$ both terms on the right
hand side converge to zero due to the fact that 
$e^{is H}W_\eta(s)$ is convergent in the operator
norm, $Ce^{-itA}P_{\rm ac}(A)$ converges strongly to $0$ for any $A$ selfadjoint and $C$ compact \cite[Lem.1,I $\mathsection$4.4]{Yafaev} and using the fact that $\rho_{eq}(\overset{\circ}{H})
E_{\rm pp}(\overset{\circ}{H})$ is compact and the following result (see $\mathsection$~\ref{SSec-H0} for the proof):

\vskip2mm
\begin{proposition}\label{P-0}%
\ For any continuous function $\Phi\in C(\mathbb{R})$ which tends to
 zero to infinity, 
we have that $\Phi(H)-\Phi(\overset{\circ}{H})$ is a compact operator.
\end{proposition}

\vspace{0.5cm}

Up to now we have shown that the limit in \eqref{hc3prima1} must equal:
  \begin{align}\label{hc3priima1}
\underset{\eta\searrow 0} {s-\lim}\left\{\underset{s\searrow-\infty} {s-\lim}\; 
E_{\rm ac}(K(1))W_\eta^*(s)E_{\rm ac}(H)\rho_{eq}(\overset{\circ}{H})
E_{\rm ac}(\overset{\circ}{H})E_{\rm ac}(H)W_\eta(s)E_{\rm
  ac}(K(1))\right\}.
\end{align}
For the next step we will need a comparison dynamics for $W_\eta(t)$, generated
by the operator with internal Dirichlet walls. To the decoupled
Hamiltonian we can associate:
\begin{equation}\label{dec-t-dep-hamilt}
\overset{\circ}{K}(\chi(\eta t)):=\overset{\circ}{H}+\chi(\eta t)V.
\end{equation}
The associated evolution
$\overset{\circ}{W}_{\eta}(t)$ is defined as the solution of the following Cauchy problem:
$$
\left\{
 \begin{array}{l}
  i\partial_t\overset{\circ}{W}_{\eta}(t)=\overset{\circ}{K}_\eta(t)\overset{\circ}
{W}_{\eta}(t)\\
  \overset{\circ}{W}_{\eta}(0)=1
 \end{array}
\right.
$$
 (its existence results by arguments similar to 
those concerning the existence of $W_{\eta}(t)$).

An important observation is the fact that $\overset{\circ}{\Delta}_D$ 
commutes with $V$ so that we have
\begin{equation}\label{finalprima5}
 \overset{\circ}{W}_{\eta}(t)=e^{-it\overset{\circ}{H}}
\left[1+\Pi_-\left(e^{-iv_-\int_0^t\chi(\eta u)du}-1\right)+
\Pi_+\left(e^{-iv_+\int_0^t\chi(\eta u)du}-1\right)\right]
\end{equation}
with the exponentials in the second factor being just complex
numbers. All terms commute which each other. 
Therefore the limit in \eqref{hc3priima1} must equal: 
\begin{equation}\label{hc1prima7}\underset{\eta\searrow 0} {s-\lim}
\left\{\underset{s\searrow-\infty} {s-\lim}\,E_{\rm
    ac}(K(1))W_\eta^*(s)E_{\rm
    ac}(H)\overset{\circ}{W}_\eta(s)\rho_{eq}(\overset{\circ}{H})
E_{\rm ac}(\overset{\circ}{H})\overset{\circ}{W}^*_\eta(s)E_{\rm
  ac}(H)W_\eta(s)
E_{\rm ac}(K(1))\right\}.
\end{equation}

We state a result which will be proved later ( see $\mathsection$~\ref{SSec-H0}):
\vskip2mm
\begin{proposition}\label{P-2}%
\ The following limits  exist in the strong operator topology:
\begin{equation}\label{hc3prima5}
\Xi_\eta:=\underset{s\searrow-\infty}{\lim}\; E_{\rm
  ac}(K(1))W_\eta^*(s)E_{\rm ac}(H)\overset{\circ}{W}_\eta(s)E_{\rm
  ac}(\overset{\circ}{H}).
\end{equation}

\end{proposition}%

\vspace{0.5cm}
One can see that the product of operators in the limit
\eqref{hc3prima5} coincides with the product of operators placed at the left of
$\rho_{eq}(\overset{\circ}{H})$ in \eqref{hc1prima7}. At the same
time, at the right of $\rho_{eq}(\overset{\circ}{H})$ is the adjoint
of the same product.

Now if we can prove that $\Xi_\eta^*$  can be written in the following  way:
\begin{equation}\label{Cor-P-2}
\Xi_\eta^*=\underset{s\searrow-\infty}{s-\lim}\;E_{\rm
  ac}(\overset{\circ}{H})\overset{\circ}{W}_\eta^*(s)E_{\rm
  ac}(H)W_\eta(s)E_{\rm ac}(K(1)),
\end{equation}
then the limit $s\to-\infty$ in \eqref{hc1prima7} would give:
\begin{equation}\label{hdcprima1}
\Xi_\eta\rho(\overset{\circ}{H})\Xi_\eta^*.
\end{equation}

Indeed, since Proposition \ref{P-2} implies the existence of the weak 
limit:
$$
\Xi_\eta^*=\underset{s\searrow-\infty}{w-\lim}\;E_{\rm
  ac}(\overset{\circ}{H})\overset{\circ}{W}_\eta^*(s)E_{\rm ac}(H)W_\eta(s)E_{\rm ac}(K(1)),
$$
then \eqref{Cor-P-2} holds if we can prove the existence of a strong
limit. Now in order to prove that a strong limit
exists, let us insert some operators in the following way:
\begin{align}
&E_{\rm
  ac}(\overset{\circ}{H})
\overset{\circ}{W}_\eta^*(s)E_{\rm ac}(H)W_\eta(s)E_{\rm
  ac}(K(1))\nonumber \\
&=E_{\rm ac}(\overset{\circ}{H})\overset{\circ}{W}_\eta^*(s)e^{-is H}
E_{\rm ac}(H)\left \{e^{is H}W_\eta(s)\right\}E_{\rm
  ac}(K(1))\nonumber \\
&=E_{\rm ac}(\overset{\circ}{H})\left
  \{\overset{\circ}{W}_\eta^*(s)e^{-is \overset{\circ}{H}} \right\}
\left \{e^{is \overset{\circ}{H}}e^{-is H}
E_{\rm ac}(H)\right\} \left \{e^{is H}W_\eta(s)\right\}E_{\rm
  ac}(K(1)).
\end{align}

Let us investigate each curly bracket. The couple $e^{is H}W_\eta(s)$
converges in norm to $\omega_\eta^*$ when
$s\to-\infty$. The factor 
$\overset{\circ}{W}_\eta^*(s)e^{-is \overset{\circ}{H}}$ converges in
norm too, see
\eqref{finalprima5}.  Finally, the factor
$e^{is \overset{\circ}{H}}e^{-is H}E_{\rm ac}(H)$ converges strongly
to the wave operator associated to the pair of Hamiltonians
$\{\overset{\circ}{H},H\}$ as stated by the following proposition
which we will prove later: 
\begin{proposition}\label{P-1}%
The wave operator
$\omega_-:=\underset{s\searrow-\infty}{s-\lim}e^{is\overset{\circ}{H}}e^{-isH}E_{ac}(H)$
exists  as a unitary map from $E_{ac}(H)\mathcal{H}$ onto $E_{ac}(\overset{\circ}{H})\mathcal{H}$ and one has:
$$
\underset{s\searrow-\infty}{s-\lim}e^{isH}e^{-is\overset{\circ}{H}}E_{ac}(\overset{\circ}{H})=\omega_-^*=E_{ac}(H)\omega_-^*.
$$
\end{proposition}%

\vspace{0.5cm}

Now we can introduce \eqref{hc3prima5} and \eqref{Cor-P-2} in
\eqref{hc1prima7}, and see that the contribution coming from the
continuous part of the spectrum will be:
\begin{equation}\label{hc1prima11}
\underset{\eta\searrow 0}{s-\lim}\; \Xi_\eta\rho(\overset{\circ}{H})\Xi_\eta^*.
\end{equation}
\\
The next step in our strategy is to prove that the strong limits of 
$\Xi_\eta$ and $\Xi_\eta^*$ exist when $\eta\searrow 0$, and they will 
equal the wave operators associated to the pair of Hamiltonians 
$\{\overset{\circ}{K}(1),K(1)\}$. First, we need to be sure that these
limiting operators exist and are complete, and this is stated by the
following proposition:
\begin{proposition}\label{P-3}
\begin{enumerate}
 \item For any $\kappa\in[0,1]$ we have 
$E_{ac}(\overset{\circ}{K}(\kappa))=E_{ac}(\overset{\circ}{H})$.
\item The following limits exist: 
\begin{align}\label{hc4prima1}
&\underset{s\searrow-\infty}{s-\lim}e^{isK(1)}e^{-is\overset{\circ}{K}(1)}
E_{ac}(\overset{\circ}{H})
=:\Xi_{0}=E_{ac}(K(1))\Xi_{0}E_{ac}(\overset{\circ}{H});\nonumber \\
&\underset{s\searrow-\infty}{s-\lim}e^{is\overset{\circ}{K}(1)}e^{-isK(1)}E_{ac}(K(1))=
\Xi_{0}^*=E_{ac}(\overset{\circ}{H})\Xi_{0}^*E_{ac}(K(1)).
\end{align}
Thus the wave operators associated to the pair
$\{\overset{\circ}{K}(1), K(1)\}$ exist and are complete.
\end{enumerate}
\end{proposition}

\vspace{0.5cm}

The next technical result establishes the adiabatic limit for the wave 
operators $\Xi_\eta$; note that Dollard \cite{Dol} investigated a
related problem in the case of short range and relatively bounded perturbations.  
\begin{proposition}\label{P-4}
$\Xi_{\eta}$ has a strong limit when $\eta\searrow 0$ and moreover 
$
\underset{\eta\searrow0}{s-\lim}\;\Xi_{\eta}=\Xi_{0}
$, where $\Xi_0$ is the stationary wave operator  associated to the pair 
$\{\overset{\circ}{K}(1),K(1)\}$ and is unitary as a map from $E_{\rm
  ac}(\overset{\circ}{H})$ onto $E_{\rm
  ac}(K(1))$..
\end{proposition}

\vspace{0.5cm}

We see that the very last thing to be shown in order to finish the
computation of the adiabatic limit in \eqref{hc1prima11}, is the
strong convergence of
$\Xi_{\eta}^*$ to $\Xi_{0}^*$ when $\eta\searrow 0$. 
Due to the completeness of the wave operator $\Xi_{0}$ (point (2) in
Proposition \ref{P-3}), we have that
$\Xi_{0}^*:E_{ac}(K(1))\mathcal{H}\to E_{ac}(\overset{\circ}{H})\mathcal{H}$ is a unitary operator. Then:
$$
\left\|\left[\Xi_{0}^*-\Xi_{\eta}^*
  \right]f\right\|_{\mathcal{H}}^2\leq 
2\|f\|_{\mathcal{H}}^2-2\Re\big(\left<\Xi_{\eta}\Xi_{0}^*f,f
\right>\big)\underset{\eta\searrow 0}{\rightarrow} 2\|f\|_{\mathcal{H}}^2-2\Re\big(\left<\Xi_{0}\Xi_{0}^*f,f
\right>\big)= 0
$$
for any $f\in E_{ac}(K(1))\mathcal{H}$ and thus we have strong
convergence of $\Xi_{\eta}^*$ to $\Xi_{0}^*$ when $\eta\searrow 0$ on
$E_{ac}(K_1)\mathcal{H}$. 

With this, the proof of the adiabatic limit in \eqref{finalprima6} is concluded.

\vspace{1cm}

The next sections of the paper are devoted to the proofs of the above
stated Propositions \ref{P-Discr}- \ref{P-4} 
and Corollary \ref{Cor-P-discr}.

\section{Absence of singular continuous spectrum}
\setcounter{equation}{0}

We give here the proof of the absence of the singular continuous spectrum for $K(\kappa)$ by establishing a limiting absorption principle. The main technical result of this section is the following lemma:
\vskip2mm
\begin{lemma}\label{LAP-K1}%
\ Let $\kappa\in [0,1]$.  There
exists a discrete set $\mathfrak{N}\subset \mathbb{R}$ such that for
any closed interval $I\subset \mathbb{R}_+\setminus\mathfrak{N}$ we have the 
estimate (here $\langle x \rangle :=\sqrt{x^2+1}$): 
\begin{equation}\label{cdh32}
\underset{z\in\{x+iy|x\in I, 0<y<\delta\}}{\sup}\left\|
e^{-\langle Q_1\rangle} R_\kappa(z)e^{-\langle Q_1\rangle} \right\|\leq C(I,\delta,\kappa)<\infty.
\end{equation}
In particular, $K(\kappa)$ has no singular continuous spectrum.
\end{lemma}%

\begin{proof}
We use geometric perturbation theory. Let us define a quadratic
partition of unity in the following way: 
$$
\chi_-^2+\chi_0^2+\chi_+^2=1,\quad\chi_{\pm}\in\,
C^\infty(\mathbb{R}),\quad\chi_{\pm}(x)=1\text{ for } \pm x>2a, \quad\chi_{\pm}(x)=0\text{ for } |x|<a
$$
$$
\chi_{0}\in\, C^\infty(\mathbb{R}),\quad\chi_{0}(x)=0\text{ for }
|x|>2a,\quad\chi_{0}(x)=1\text{ for } |x|<a. 
$$
Fix some $L>2a$. Introduce the operator $K_{\kappa,L}$ obtained from
$K(\kappa)$ on the region $\L\cap (-L,L)$ by imposing Dirichlet
boundary conditions at $x=\pm L$. The operator $K_{\kappa,L}$ has
compact resolvent, and let us denote it with $R_{\kappa,L}(z)$. Here $z\in
\mathbb{C}\setminus \sigma(K_{\kappa,L})$. 

Now let us define an approximation for $R_\kappa(z)$ by the following formula:
$$
\widetilde{R}_\kappa(z):=\chi_-(Q_1)\overset{\circ}{R}_\kappa(z)\chi_-(Q_1)+\chi_0(Q_1)
R_{\kappa,L}(z)\chi_0(Q_1)+\chi_+(Q_1)\overset{\circ}{R}_\kappa(z)\chi_+(Q_1).
$$
Note that on the support of $\chi_{\pm}(Q_1)$ the differential
operators $K(\kappa)$ and 
$\overset{\circ}{K}(\kappa)$ coincide, while on the support of
$\chi_0(Q_1)$ the operators $K(\kappa)$ and $K_{\kappa,L}$ coincide,  so that we can write
\begin{align}\label{cdh33}
(K(\kappa)-z)\widetilde{R}_\kappa(z)&={\rm Id}+
\left[\overset{\circ}{H},\chi_-(Q_1)
\right]\overset{\circ}{R}_\kappa(z)\chi_-(Q_1)
+\left[H,\chi_0(Q_1) \right]R_{\kappa,L}(z)\chi_0(Q_1)\nonumber \\ 
&+\left[\overset{\circ}{H},\chi_+(Q_1) \right]\overset{\circ}{R}_\kappa(z)\chi_+(Q_1).
\end{align}
The above commutators are first order differential operators:
\begin{align}\label{cdh34}
\left[\overset{\circ}{H},\chi_{\pm}(Q_1)
\right]&=-2i\chi'_{\pm}(Q_1)P_1-\chi''_{\pm}(Q_1),
\nonumber \\
\left[H,\chi_0(Q_1)\right] &=-2i\chi'_{0}(Q_1)P_1-\chi''_{0}(Q_1).
\end{align}
Thus \eqref{cdh33} can be put in the following form:
\begin{align}\label{cdh35}
(K(\kappa)-z)\widetilde{R}_\kappa(z)&={\rm Id}+ X(z),\qquad  
e^{\langle Q_1\rangle}X(z)\in\mathbb{B}(\mathcal{H}),
\end{align}
where the boundedness of $e^{\langle Q_1\rangle}X(z)$ is due to the compact support of the
functions appearing on the left-hand side of the operator $X(z)$. Thus
we can write the identity:
$$
e^{-\langle Q_1\rangle}R_\kappa(z)e^{-\langle Q_1\rangle}=
e^{-\langle Q_1\rangle}\widetilde{R}_\kappa(z)e^{-\langle Q_1\rangle}-
e^{-\langle Q_1\rangle}R_\kappa(z)X(z)e^{-\langle Q_1\rangle}.
$$
Since for large values of ${ \Im}(z)$ the norm of $e^{\langle Q_1\rangle}
X(z)$ tends to $0$, we can write at least for those values of $z$ that:
$$
e^{-\langle Q_1\rangle}R_\kappa(z)e^{-\langle Q_1\rangle}=e^{-\langle
  Q_1\rangle}\widetilde{R}_\kappa(z)e^{-\langle Q_1\rangle}\left[1+
e^{\langle Q_1\rangle}X(z)e^{-\langle Q_1\rangle}\right]^{-1}.
$$ 
Now $e^{\langle Q_1\rangle} X(z)e^{-\langle Q_1\rangle}$ is
compact and analytic in the upper complex plane, and has a bounded
limit from above on any interval $I$ which avoids the
discrete set of thresholds in the leads and the discrete spectrum of
$K_{\kappa,L}$. Moreover, due to the exponential decaying weight on the right
and the compactly supported cut-offs on the left, $e^{\langle Q_1\rangle} X(z)e^{-\langle Q_1\rangle}$
can be analytically continued to the set
$\{x+iy|x\in I, -\delta <y<\delta\}$ for $\delta$ small enough. 
Thus we can apply the analytic Fredholm alternative on
this set and conclude that 
$\left[1+e^{\langle Q_1\rangle} X(z)e^{-\langle
    Q_1\rangle}\right]^{-1}$ exists on $I$ outside a discrete set of
points. 
\end{proof}

\section{Adiabatic limit of the discrete subspace}\label{Discr-sp}%

\setcounter{equation}{0}

In order to simplify our presentation, we adopt the 
conditions of Hypothesis \ref{Hyp-33} which means that we have 
$N=2$ discrete eigenvalues which might cross at only one point when $\kappa$
varies. Moreover, they remain well isolated from the continuous
spectrum. Under these conditions, Rellich's Theorem 
(Theorem II.61 in \cite{K}) states that the two eigenvalues are 
given by two real analytic functions
$\{\varepsilon_j(\kappa)\}_{j\in\{1,2\}}$ defined for
$\kappa\in[0,1]$. If they cross at $\kappa_0\in (0,1)$ and only there,
then there must exist two constants $C>0, M\in\mathbb{N}^*$ such that 
\begin{equation}\label{hcprima1}
|\varepsilon_1(\kappa)-\varepsilon_2(\kappa)|\geq C
|\kappa-\kappa_0|^M,\quad \kappa\in [0,1].
\end{equation}
Moreover, their corresponding orthogonal projections $E_j(\kappa)$ can also be
chosen to be real analytic on $[0,1]$. 

\subsection{Proof of Proposition \ref{P-Discr}}

Let us focus on $j=1$. We will have to show the equality: 
\begin{equation}\label{hcprima2}
E_1(1)=\underset{\eta\searrow
  0}{n-\lim}\left[\underset{s\searrow-\infty}{n-\lim}
B_\eta(s)\right],\quad\text{with}\quad
B_\eta(s):=
    W_{\eta}(s)^*
E_1(\chi(\eta s))W_{\eta}(s).
\end{equation}
This follows clearly from the next result.
\vskip2mm
\begin{lemma}\label{hcprima8}
We have:
\begin{equation}\label{hcprima6}
B_\eta(0)=E_1(\chi(0))=E_1(1)\quad {\rm and}\quad \lim_{\eta \searrow 0}\left \{\sup_{s\leq 0}
||B_\eta(s) -B_\eta(0)||\right \}=0.
\end{equation}
\end{lemma}
\begin{proof}
The first two equalitites are obvious. For the limit let us remember that there exists a unique 
critical time $t_0<0$ when $\chi(t_0)=\kappa_0$ which 
corresponds to the intersection of the two eigenvalues. 
Fix some $0<\delta<1$ (to be chosen later in a more precise way). 
We split the negative semi-axis  
$\mathbb{R}_-$ in three parts:
\begin{equation}\label{hcprima9}
\mathbb{R}_-=\left (-\infty, \frac{t_0-\eta^\delta}{\eta}\right ]\cup
\left [\frac{t_0-\eta^\delta}{\eta},
\frac{t_0+\eta^\delta}{\eta}\right ]\cup 
\left [\frac{t_0+\eta^\delta}{\eta},0\right ].
\end{equation}
In what follows we will investigate how $B_\eta(\cdot)$ changes when
$s$ goes through each sub-interval. 

\noindent
{\sf Near the crossing:} 
Let us first consider the interval in the middle 
$\left [\frac{t_0-\eta^\delta}{\eta},
  \frac{t_0+\eta^\delta}{\eta}\right ]$. This is the "gap-less region",
but nevertheless, it is easiest to deal with. 
From the definition of $B_\eta(s)$ in \eqref{hcprima2}, and since 
$K(\kappa)$ commutes with $E_1(\kappa)$, we have the important 
identity:
\begin{equation}\label{hcprima11}
\partial_s B_\eta(s)=\eta \chi'(\eta s)W_{\eta}^*(s)
E_1'(\chi(\eta s))W_{\eta}(s),
\end{equation}
where $E_1'(\kappa)$ is uniformly bounded in $\kappa\in
[0,1]$ due to the real analyticity of the projector. We write:
\begin{equation}\label{hcprima10}
B_\eta\left (\frac{t_0+\eta^\delta}{\eta}\right )-B_\eta\left 
(\frac{t_0-\eta^\delta}{\eta}\right )=
\int_{\frac{t_0-\eta^\delta}{\eta}}^{\frac{t_0+\eta^\delta}{\eta}}\partial_s B_\eta(s)ds.
\end{equation}
This implies:
\begin{equation}\label{hcprima12}
\left \Vert   B_\eta\left (\frac{t_0+\eta^\delta}{\eta}\right )-B_\eta\left 
(\frac{t_0-\eta^\delta}{\eta}\right )    \right \Vert \leq
C \eta^\delta.
\end{equation}

\noindent
{\sf Outside the crossing:} In the other two intervals the eigenvalue 
$\varepsilon_1(\chi(\eta s))$ is isolated from the rest of the
spectrum, as can be inferred from our Hypothesis \ref{Hyp-3} and
\ref{Hyp-33}. More precisely, let us show that it is situated at a distance
larger than $C\eta^{M\delta}$ than the rest of the spectrum. Indeed,  
using the splitting from \eqref{hcprima1} we may write 
$$|\varepsilon_1(\chi(\eta s))-\varepsilon_2(\chi(\eta s))|\geq C
|\chi(\eta s)-\chi(t_0)|^M\geq \tilde{C} \eta^{M\delta}$$
for every $s$ situated at a distance larger than $\eta^{-1+\delta}$
from $t_0/\eta$. Here $\tilde{C}>0$ can be chosen uniformly in $\eta$
because we assumed that $\chi'(t)>0$ and $|\chi''|$ is integrable. 

It means that we can find a positively oriented
simple contour $\Gamma_\eta$ which only contains
$\varepsilon_1(\chi(\eta s))$ and the following estimate holds true:
\begin{equation}\label{hcprima13}
D_\eta:=\sup_{s\in \mathbb{R}_-\setminus \left [\frac{t_0-\eta^\delta}{\eta},
  \frac{t_0+\eta^\delta}{\eta}\right ]}
 \sup_{z\in \Gamma_\eta}||(K(\chi(\eta s))-z)^{-1}||\leq C
\eta^{-M\delta}.
\end{equation}
We can choose the length of the contour $\Gamma_\eta$ to be of order $1/D_\eta$.
We will treat this region by using a second order adiabatic
development for the 
quasi-eigenprojector given by the adiabatic theory (see \cite{N2,
  teufel} and references therein). If 
$$
X(\kappa):=\left[E^\bot_1(\kappa)E^\prime_1(\kappa)
  E_1(\kappa)-E_1(\kappa)E^\prime_1(\kappa) 
E^\bot_1(\kappa)\right],
$$
we define: 
\begin{align}\label{hcprima14}
&F_{\eta}(s):=B_\eta(s)+\eta \chi'(\eta s) W_{\eta}^*(s)Y(\chi(\eta s))W_{\eta}(s)\\
&Y(\chi(\eta s)):=-\frac{1}{2\pi}\oint_{\Gamma_\eta}dz\,
\big(K(\chi(\eta s))-z\big)^{-1}X(\chi(\eta s))\big(K(\chi(\eta
s))-z\big)^{-1},\nonumber 
\end{align}
where the operator $Y(\kappa)$ is a solution to the commutator
equation $i[K(\kappa),Y(\kappa)]=-E^\prime_1(\kappa)$. 
The operator $F_{\eta}(s)$ is constructed in such way that when we compute
$\partial_s F_{\eta}(s)$, the term $\partial_s B_{\eta}(s)$ gets
canceled and we have the identity:
\begin{align}\label{hcprima15}
&\partial_s F_{\eta}(s)=\\
&-\eta
^2W_{\eta}^*(s)\left \{\partial_x \frac{\chi'(x)}{2\pi}\oint_{\Gamma_\eta}dz\,\big(K(\chi(x))-z\big)^{-1}
X(\chi(x))\big(K(\chi(x))-z\big)^{-1}\right\}_{{x=\eta s}} W_{\eta}(s).\nonumber 
\end{align}
Note that $Y(\chi(\eta s))$ is a bounded operator obeying 
\begin{equation}\label{cdh2}
 \Vert Y(\chi(\eta s))\Vert \leq C \eta^{-M\delta}, \quad s\in \mathbb{R}_-\setminus \left [\frac{t_0-\eta^\delta}{\eta},
  \frac{t_0+\eta^\delta}{\eta}\right ],
\end{equation}
which is a consequence of \eqref{hcprima13} and because our choice of the contour $\Gamma_\eta$. It follows that we can
write a rough bound of the type
\begin{equation}\label{hcprima16}
||\partial_s F_{\eta}(s)||\leq C
\eta^{2-2 M\delta}(|\chi'(\eta s)|^2+|\chi''(\eta s)|),\quad s\in \mathbb{R}_-\setminus \left [\frac{t_0-\eta^\delta}{\eta},
  \frac{t_0+\eta^\delta}{\eta}\right ].
\end{equation}
Thus on any sub-interval $[s_1,s_2]$ of the negative real axis where the above estimate holds true we can write: 
\begin{equation}\label{cdh3}
||F_{\eta}(s_1)-F_{\eta}(s_2)||\leq C \eta^{1-2M\delta},
\end{equation}
due to the integrability properties of $\chi$ (see \eqref{cdh1}). 
From \eqref{hcprima14} and \eqref{cdh2} we can derive the estimate:
\begin{equation}\label{hcprima17}
||B_\eta(s)-F_{\eta}(s)||\leq C
\eta^{1- M\delta}, \quad s\in \mathbb{R}_-\setminus \left [\frac{t_0-\eta^\delta}{\eta},
  \frac{t_0+\eta^\delta}{\eta}\right ].
\end{equation}
Up to a use of the triangle inequality, on any sub-interval $[s_1,s_2]$ of $\mathbb{R}_-\setminus \left [\frac{t_0-\eta^\delta}{\eta},
  \frac{t_0+\eta^\delta}{\eta}\right ]$ we can write:
\begin{align}\label{hcprima18}
||B_\eta(s_1)-B_\eta(s_2)||&\leq 
 C
\eta^{1-2 M\delta}.
\end{align}
This estimate together with \eqref{hcprima12} imply:
$$ ||B_\eta(s)-B_\eta(0)||=||W_\eta^*(s)E_1(\chi(\eta s))W_\eta(s)-E_1(1)||\leq C(\eta^\delta +\eta^{1-2 M\delta}),\quad s\leq 0.$$
Choose now any $\delta \in (0,1/(2M))$; then \eqref{hcprima6} is proved, which concludes the lemma. \end{proof}

\subsection{Proof of Corollary \ref{Cor-P-discr}}

The limit in \eqref{hc2prima1} is a trivial consequence of \eqref{E-1} and the result of Proposition \ref{P-Discr}. The proof of \eqref{hc2prima2} is a bit longer. We start with a lemma:
\begin{lemma}\label{L-discr-1}%
\ At fixed $\eta>0$, we have the limit $
\underset{s\searrow-\infty}{n-\lim}B_\eta(s)= \omega_\eta E_1(0)\omega_\eta^*$.
\end{lemma}%
\begin{proof}
 This can be seen by
  writing 
$$B_\eta(s)= \{W_{\eta}^*(s)e^{-isH}\}\{e^{isH}E_1(\chi(\eta
s))e^{-isH}\}
\{e^{isH}W_{\eta}(s)\}\underset{s\searrow-\infty}{\to} \omega_\eta E_1(0)\omega_\eta^* $$
where we used the fact that each parenthesis converges in norm (even though
not uniformly in $\eta$). \end{proof}
\vspace{0.2cm}
\begin{corollary}\label{hcprima8-a} (of Lemmas \ref{hcprima8} and \ref{L-discr-1})
\begin{equation}
\lim_{\eta \searrow 0}\left \Vert 
\omega_\eta E_1(0)\omega_\eta^* -E_1(1)\right \Vert =0.
\end{equation}
\end{corollary}

\vspace{0.2cm}

We can now prove \eqref{hc2prima2}:
\begin{align}\label{hc2prima22}
&\underset{\eta\searrow0}{n-\lim}
\left[\underset{s\searrow-\infty}{n-\lim}
e^{isH}E_{ac}(H)W_{\eta}(s)E_{pp}(K(1))\right]\nonumber \\
&=\underset{\eta\searrow0}{n-\lim}\left[\underset{s\searrow-\infty}{n-\lim}\;
\left
  \{e^{isH}W_{\eta}(s)\right\}W_{\eta}^*(s)\{1-E_{pp}(H)\}W_{\eta}(s)E_{pp}(K(1))\right]\nonumber \\
&=\underset{\eta\searrow0}{n-\lim}\left[\omega_\eta^*\;\underset{s\searrow-\infty}{n-\lim}\;
W_{\eta}^*(s)\{1-E_{pp}(H)\}W_{\eta}(s)E_{pp}(K(1))\right]\nonumber \\
&=\underset{\eta\searrow0}{n-\lim}\left[\omega_\eta^*\;\{E_{pp}(K(1))-\omega_\eta E_{pp}(K(0))\omega_\eta^*\}\right] E_{pp}(K(1))=0,
\end{align}
where we 
used Corollary \ref{hcprima8-a}, and the fact that 
$\omega_\eta^*=\underset{s\searrow-\infty}{n-\lim}e^{isH}W_{\eta}(s)$ has norm one.

\section{Existence and completeness of stationary wave operators}\label{sc-w-bias}

\setcounter{equation}{0}

In this Section we analyze the pair of Hamiltonians
$\{\overset{\circ}{H}+\kappa V,H+\kappa V\}$ and prove Propositions 
\ref{P-0}, \ref{P-1}
and \ref{P-3}.
\subsection{Proof of Proposition \ref{P-0}}\label{SSec-H0}%

We can approximate the function $\Phi$ in the uniform norm with a sequence of 
$C_0^\infty(\mathbb{R})$ functions $\Phi_n$. Thus if we can prove the
proposition for smooth and compactly supported functions, then we are
done. For such a $\Phi_n$ we can apply for example the
Helffer-Sj\"ostrand formula (or any other norm convergent functional
calculus involving the resolvent) and argue that we can approximate in norm
the difference $\Phi_n(H)-\Phi_n(\overset{\circ}{H})$ with a linear
combination of differences of resolvents of the type 
\begin{equation}\label{cdh10}
\sum_{j=1}^{N} C_j \left
  \{(H-z_j)^{-1}-(\overset{\circ}{H}-z_j)^{-1}\right \},
\end{equation}  
where $C_j$ are complex coefficients and $z_j$ are
complex numbers with nonzero imaginary part. Thus one can reduce the
problem to showing that 
$$
(H-z)^{-1}-(\overset{\circ}{H}-z)^{-1}=:\,R(z)-\overset{\circ}{R}(z)
$$ is
compact for some $z$ with $\Im(z)>0$. 

Our decoupled Hamiltonian $\overset{\circ}{H}$ (see
\eqref{decoupl-domains}-
\eqref{dec-hamilt}) is a direct sum of three commuting operators, and
we have $\sigma_{sc}\big(\overset{\circ}{H}\big)=\emptyset$ and
$$
\sigma_{pp}\big(\overset{\circ}{H}\big)=\sigma_{pp}\big(\Pi_0\overset{\circ}{H}\Pi_0\big)=
\sigma\big(\Pi_0\overset{\circ}{H}\Pi_0\big)\subset\mathbb{R}_+.
$$

Let us denote by $\{w_n\}_{n\in\mathbb{N}}$ the complete orthonormal
set 
of eigenvectors of $\mathfrak{L}_{\D}$ in $L^2(\D)$ (see
(\ref{sect-Lapl})), 
having eigenvalues $\{\lambda_n\}_{n\in\mathbb{N}}$ so that 
$\sigma_{pp}\big(\mathfrak{L}_{\D}\big)=\{\lambda_n\}_{n\in\mathbb{N}}$;
let $P_n$ be 
the 1-dimensional orthogonal projection on $w_n$ in $L^2(\D)$. In particular, 
$$
\sigma_{ac}\big(\overset{\circ}{H}\big)=\sigma_{ac}\big(\Pi_-\overset{\circ}{H}\Pi_-\oplus\Pi_+\overset{\circ}{H}\Pi_+\big)=\sigma\big(\Pi_-\overset{\circ}{H}\Pi_-\oplus\Pi_+\overset{\circ}{H}\Pi_+\big)=[\lambda_1,\infty).
$$

Then for $z\in\mathbb{C}\setminus[0,\infty)$ we have
$$
\overset{\circ}{R}(z)=\underset{n\in\mathbb{N}}{\oplus}
\left[\big(\mathfrak{l_-}-(z-\lambda_n)\big)^{-1}\pi_-\oplus
\big(\mathfrak{l_+}-(z-\lambda_n)\big)^{-1}\pi_+\right]P_n\,\oplus\,(-\overset{\circ}\Delta_D+w-z)^{-1}
$$
with $\pi_{\pm}:L^2(\mathbb{R})\rightarrow L^2(\I_{\pm})$ the
orthogonal projections and $\big(\mathfrak{l}-(z-\lambda_n)\big)^{-1}$
the resolvent of the longitudinal kinetic energy on $\I_{\pm}$ with
Dirichlet conditions at $\pm a$.

In order to study the Hamiltonian $H$ and its relation with $\overset{\circ}{H}$, let us first observe that they are two self-adjoint extensions of the same symmetric operator 
$$
D_{0}:=-\Delta_D+w:C^\infty_0({\L_-}\cup{\C}\cup{\L_+})\rightarrow\mathcal{H}.
$$ 
Let $D_{0}^*$ be the adjoint of this symmetric operator.
In order to compare the two resolvents, $R(z)$ and
$\overset{\circ}{R}(z)$ for $z\in\mathbb{C}\setminus[0,\infty)$, we
note that $D_{0}^*$ extends both self-adjoint operators $H$ and
$\overset{\circ}{H}$ so that:  
$$
\left(R(z)- \overset{\circ}{R}(z)\right)\mathcal{H}\subset \ker \left(D_{0}^*-z\right).
$$
 Notice that for $u\in\ker\left(D_{0}^*-z\right)$
 the distribution $D_{0}^* u-zu$ has support in the part of the boundary
 $\D_-\cup\D_+$, where $\D_\pm:=\L_\pm\cap\big(\{\pm a\}\times\mathbb{R}^d\big)$; thus on $\L_-\cup\L_+$ they satisfy the equation:
$$
-\overset{\circ}{\Delta}_{D,\pm} u_\pm=z u_\pm
$$
with the boundary condition $\left.u_{\pm}\right|_{\I_{\pm}\times\partial\D}=0$, for $u_{\pm}:=\left.u\right|_{\L_{\pm}}$.
Then standard arguments show that our vectors
$u_{\pm}\in\mathcal{H}_{\pm}$ 
must be of the form
$
u_{\pm}=\underset{n\in\mathbb{N}}{\oplus}\alpha_n\big(u_{\pm,n}\otimes w_n\big),
$
where the functions 
$u_{\pm,n}\in L^2(\I_{\pm})$ satisfy the equation
$\mathfrak{l}_{\pm}u_{\pm,n}=(z-\lambda_n)u_{\pm,n}$. Thus 
$u_{\pm,n}=\beta_{\pm,n}e^{\zeta_{\pm,n}x}$ with $\zeta_{\pm,n}$ the
unique complex square root of 
$z-\lambda_n$ having $\pm Re\zeta_{\pm,n}<0$.
Let us observe that due to the fact that $z\in\mathbb{C}\setminus[0,\infty)$ and $\lambda_n>0$, our sequence  $\{\left|Re\zeta_{\pm,n}\right|\}$ contains strictly positive numbers, and moreover, diverges with $n$. Thus the infimum below is positive:
\begin{equation}\label{gamma0}
\gamma_0(z):=\underset{n\in\mathbb{N}}{\inf}\left|Re\zeta_{\pm,n}\right|>0.
\end{equation}
We have thus proved the following statement (here $Q_1$ is the multiplication operator with the longitudinal coordinate):
\vskip2mm
\begin{lemma}\label{exp-dec}%
\ Let $z\in\mathbb{C}\setminus[0,\infty)$ and $\gamma_{\pm}\in \mathbb{R}_+\setminus\{0\}$ be such that $0<\gamma_{\pm}<\gamma_0(z)$ (defined in \ref{gamma0}), then we have the following estimations of exponential decay:
\begin{equation}\label{L2-descr-exp-0}
\left\|e^{\pm\gamma_{\pm} Q_1}\Pi_\pm\big(R(z)-\overset{\circ}{R}(z)\big)e^{\pm\gamma_{\pm} Q_1}\Pi_\pm\right\|\leq c,
\end{equation} 
and for $\Psi_\alpha(x):=e^{\alpha \sqrt{x^2+1}}$ (with $0\leq\alpha<\gamma_0(z)$) we have:
\begin{equation}\label{L2-descr-exp}
\left\|\Psi_\alpha(Q_1)\big(R(z)-\overset{\circ}{R}(z)\big)\Psi_\alpha(Q_1)\right\|\leq c.
\end{equation} 
\end{lemma}%
\noindent
Taking into account that $R(z)\mathcal{H}$ and
$\overset{\circ}{R}(z)\mathcal{H}$ are both contained in
$H^1(\L_-)\oplus H^1(\C)\oplus H^1(\L_+)$, the above estimate 
\eqref{L2-descr-exp} together with the compactness of Sobolev
embeddings for compact domains, imply that
$
R(z)-\overset{\circ}{R}(z)
$
are compact operators for any $z\in\mathbb{C}\setminus\mathbb{R}$.

\subsection{Proof of Proposition \ref{P-1}}\label{WOp-H}

The absence of singular continuous spectrum will be proved later on in
Lemma \ref{LAP-K1}; here we only show completeness of wave operators
by using the Birman-Kuroda method \cite{Yafaev}. In other words, we want to show
that the difference between some large enough powers of the 
resolvents is a trace class operator.  

We need to elaborate on the previous definition of $\Psi_{\alpha}\in
C^\infty(\mathbb{R})$ introduced in Lemma~\ref{exp-dec}. We now allow any $\alpha\in \mathbb{R}$ and 
further more: 
\begin{equation}\label{Psi}
\Psi_{\alpha}(x)=e^{\alpha \sqrt{x^2+1}},\quad \mbox{{so that}}\quad
|(\partial^s \ln \Psi_{\alpha})(x)|\leq C_s|\alpha|, \quad s\geq
1,\; x\in\mathbb{R}.
\end{equation}
We observe that $\Psi_{\alpha}(x)$ is invertible everywhere on $\mathbb{R}$, and 
if $\alpha >0$ then $\Psi_{\alpha}^{-1}\in L^k(\mathbb{R})$ for any
$k\geq 1$. Another fact we shall use here is that the following multiple commutator is bounded:
\begin{equation}\label{Hyp-1}
\left[P_1,\left[P_1,\ldots\big[P_1,w\big]\ldots\right]\right]
\end{equation}
where $P_1:=-i\partial_x$ and $x$ denotes the first (longitudinal) variable of $\mathbb{R}^{d+1}$.

\vsth
\begin{lemma}\label{cdhlemma1}%
\ Fix $z\in\mathbb{C}\setminus[0,\infty)$. Then there exist 
$k_d\in\mathbb{N}$ large enough and $\alpha(z)>0$ small
 enough such that for any $k\geq k_d$ the operators 
$\overset{\circ}{R}(z)^{k}$ and 
$\Psi_{\alpha}(Q_1)\overset{\circ}{R}(z)^{k}\Psi_{\alpha}(Q_1)^{-1}$
with $|\alpha|<\alpha(z)$ are bounded operators from $L^2(\L)$ into $BC(\L)$ (the bounded continuous functions on $\L$).
\end{lemma}%
\begin{proof} The result for $\overset{\circ}{R}(z)^{k}$ is based 
on the fact that 
\begin{equation}\label{cdh20}
\overset{\circ}{R}(z)^{k} \mathcal{H}\subset \left(H^{2k}(\L_-)\bigcap H^1_0(\L_-)\right)\oplus 
\left(H^{2k}(\C)\bigcap H^1_0(\C)\right)\oplus\left(H^{2k}(\L_+)\bigcap H^1_0(\L_+)\right)
\end{equation}
which is based on the commutator estimate in (\ref{Hyp-1}). 
Now if $k$ is large enough 
(depending only on the dimension $d$), the right hand side becomes a 
subset of $BC(\L)$.

In order to prove the same inclusion for the other operator, let us
note that we can use a Combes-Thomas type rotation \cite{combes}:
$
\Psi_{\alpha}(Q_1)\overset{\circ}{H}\Psi_{\alpha}(Q_1)^{-1}=\overset{\circ}{H}+T_{\alpha},
$
where $T_\alpha$ is a first order differential operator which has the following mapping property:
\begin{equation}\label{cdh21}
T_\alpha:H^k(\L_-)\oplus H^k(\C)\oplus H^k(\L_+)\longrightarrow H^{k-1}(\L_-)\oplus H^{k-1}(\C)\oplus H^{k-1}(\L_+).
\end{equation}
Now if $|\alpha|$ is small enough, one can prove by induction with 
respect to $k$ that 
\begin{equation}\label{cdh22}
\Vert \overset{\circ}{H}^k\Psi_{\alpha}(Q_1)\overset{\circ}{R}(z)^{k}
\Psi_{\alpha}(Q_1)^{-1}\Vert <\infty,
\end{equation}
which means:
\begin{equation}\label{cdh23}
\Psi_{\alpha}(Q_1)\overset{\circ}{R}(z)^{k}\Psi_{\alpha}(Q_1)^{-1}\mathcal{H}\subset H^{2k}(\L_-)\oplus H^{2k}(\C)\oplus H^{2k}(\L_+)
\end{equation}
and we are done. \end{proof}

\vspace{0.5cm}

\begin{corollary}\label{psiR0-HS}
Let $F\in L^2(\mathbb{R})$. Then there exists $k_d\in\mathbb{N}$
depending on the dimension $d$ such that for any
$z\in\mathbb{C}\setminus[0,\infty)$ and any $k\geq k_d$ we have that
$F(Q_1)\overset{\circ}{R}(z)^{k}$ and
$F(Q_1)\Psi_{\alpha}(Q_1)\overset{\circ}{R}(z)^{k}\Psi_{\alpha}(Q_1)^{-1}$
with $|\alpha|<\alpha(z)$, are Hilbert-Schmidt operators on $\mathcal{H}$.
\end{corollary}
\begin{proof}
Let us denote by $T$ either $\overset{\circ}{R}(z)^{k}$ or 
$\Psi_{\alpha}(Q_1)\overset{\circ}{R}(z)^{k}\Psi_{\alpha}(Q_1)^{-1}$
appearing in the previous lemma. For any $f\in \mathcal{H}$ we have
that $Tf$ is a bounded and continuous function. 
Then for any fixed ${\bf x}\in \L$, the mapping 
\begin{equation}\label{cdh24}
 \mathcal{H}\ni f\mapsto (Tf)({\bf x})\in \mathbb{C}
\end{equation}
defines a bounded linear functional on $\mathcal{H}$, 
uniformly bounded in ${\bf x}$. The Riesz representation theorem
allows us to conclude that $T$ has an integral kernel obeying
\begin{equation}\label{cdh25}
 \sup_{{\bf x}\in \L}\int_{\L}|T({\bf x},{\bf y})|^2 d{\bf y} <\infty.
\end{equation}
Thus for any function $F\in L^2(\mathbb{R})$, the operators $F(Q_1)T$ 
have integral kernels of class $L^2(\L\times\L)$, hence they are Hilbert-Schmidt.
\end{proof}

\vspace{0.5cm}

\begin{lemma}\label{psiR-HS}
Fix $z\in\mathbb{C}\setminus[0,\infty)$. Then there exist 
$k_d\in\mathbb{N}$ large enough and $\alpha(z)>0$
small enough,  such
that for any $k\geq k_d$ we have that  $F(Q_1)R(z)^{k}$ and 
$F(Q_1)\Psi_{\alpha}(Q_1)R(z)^{k}\Psi_{\alpha}(Q_1)^{-1}$ with
$|\alpha|<\alpha(z)$ are Hilbert-Schmidt operators on $\mathcal{H}$
for any measurable function $F\in L^2(\mathbb{R})$.
\end{lemma}
\begin{proof}
We use similar arguments with those
for $\overset{\circ}{H}$ but this time repeated for $H$. We do not give further details.
\end{proof}

\vspace{0.5cm}

The final technical result needed for the Birman-Kuroda theorem is the following:

\vsth
\begin{lemma}%
\ Let $z\in\mathbb{C}\setminus[0,\infty)$. Then there exists 
$n_d\in\mathbb{N}$ large enough such that for any $n\geq n_d$ we have
that $\big[R(z)^n-\overset{\circ}{R}(z)^n\big]\in\mathbb{B}_1(\mathcal{H})$ (the set of trace class operators).
\end{lemma}%
\begin{proof}
Let us fix $z\in\mathbb{C}\setminus[0,\infty)$.
 We start with the formula (valid for any $p\in\mathbb{N}$):
\begin{equation}\label{dif-rez}
R(z)^p-\overset{\circ}{R}(z)^p=\sum_{0\leq j\leq p-1}R(z)^j\big(R(z)-
\overset{\circ}{R}(z)\big)\overset{\circ}{R}(z)^{p-1-j}.
\end{equation}
Let us choose $0<\alpha<\min\{\alpha(z),\gamma_0(z)\}$, where
$\gamma_0(z)$ is the same as in Lemma \ref{exp-dec}. 
Let us choose $p\geq 4k_d+1$ and observe that in (\ref{dif-rez}) we
either have $j\geq 2k_d$ or $p-j-1\geq 2 k_d$. The idea is to prove
that operators of the type  
$$R(z)^{2k_d}
\Psi_{\alpha}(Q_1)^{-1} \quad {\rm or}\quad 
\Psi_{\alpha}(Q_1)^{-1}\overset{\circ}{R}(z)^{2k_d}$$
are trace class, which together with Lemma \ref{exp-dec} would finish
the proof. Indeed, we can write:
\begin{align}\label{cdh30}
R(z)^{2k_d}\Psi_{\alpha}(Q_1)^{-1}&=\left \{ \Psi_{-\alpha/2}(Q_1)
[\Psi_{\alpha/2}(Q_1)R(z)^{k_d}
\Psi_{-\alpha/2}(Q_1)]\right\}\nonumber \\
&\cdot  \left\{ \Psi_{-\alpha/2}(Q_1)
[\Psi_{\alpha}(Q_1)R(z)^{k_d}\Psi_{-\alpha}(Q_1)]\right
\},
\end{align}
where the right hand side is - according to Lemma
\ref{psiR-HS} - a product of two Hilbert-Schmidt operators. The other
operator can be treated in a similar way, up to taking the
adjoint. The proof is over.    
\end{proof}

\subsection{Proof of Proposition \ref{P-3}}

We now want to study the pair of Hamiltonians 
$K(\kappa)=H+\kappa V$ and
$\overset{\circ}{K}(\kappa)=\overset{\circ}{H}+\kappa V$, for any
$\kappa\in[0,1]$ and prove Proposition \ref{P-3}. The only difficulty 
comes from the fact that the perturbation $V$ has a singular 
commutator with $H$ and, at the same time, it does not tend to zero at
infinity. 

But for $\overset{\circ}{K}(\kappa)$ there is no difficulty due to the
fact that $\overset{\circ}{H}$ commutes with the bias $V$, while the
last one is just a multiple of the identity on each orthogonal
subspace in the decomposition
$\mathcal{H}=\mathcal{H}_-\oplus\mathcal{H}_0 \oplus\mathcal{H}_+$. 

In fact, the only result which cannot be obtained just like in the
previous section is the following lemma:

\vsth
\begin{lemma}\label{psi-K1-HS}%
\ Fix $z\in\mathbb{C}\setminus[0,\infty)$. 
There exists $k'_d\in\mathbb{N}$ large enough such that for any any 
$k\geq k'_d$ we have that the operators  $F(Q_1)R_\kappa(z)^{k}$ and $F(Q_1)
\Psi_{\alpha}(Q_1)R_\kappa(z)^{k}\Psi_{\alpha}(Q_1)^{-1}$ with
$|\alpha|<\alpha(z)$ (here $\Psi_\alpha$ is as in (\ref{Psi})) are 
Hilbert-Schmidt operators on $\mathcal{H}$ for any function 
$F\in L^2(\mathbb{R})$.
\end{lemma}%
\begin{proof} 
The perturbation $V$ is still relatively bounded with respect to $H$
with $0$ relative bound, but its commutator with $H$ defined as a 
sesquilinear form on the domain of $H$ is singular. We observe that
the main difficulty comes from the fact that the range of the
operator $R_\kappa(z)^k$ is no longer contained in the Sobolev space 
$H^{2k}(\L)$ but only in $H^2(\L)$ for any $k\in\mathbb{N}$, due to
the singularity of the commutator of $-\Delta$ with $V$. In fact the 
situation is a bit better due to the fact that $V$ 
(being constant in the $\D$-space) commutes with all derivatives
with respect to directions from $\D$. Thus, using the results in \cite{LM2}, 
we may conclude that:
$$
R_\kappa(z)^k\mathcal{H}\subset H^2\big(\mathbb{R};H^{2(k-1)}(\mathbb{R}^d)\big)\cap L^2(\L)\subset BC\big(\mathbb{R};H^{2(k-1)}(\mathbb{R}^d)\big)\cap L^2(\L)\subset
$$
$$
\subset BC\big(\mathbb{R};BC(\mathbb{R}^d)\big)\cap L^2(\L)\subset BC(\L)
$$
for $k\geq k^\prime_d$ depending only on the dimension $d$. Thus the
proof goes on exactly as in Section \ref{sc-w-bias} and we are done
with $F(Q_1)R_\kappa(z)^{k}$. 

Regarding $F(Q_1)
\Psi_{\alpha}(Q_1)R_\kappa(z)^{k}\Psi_{\alpha}(Q_1)^{-1}$, we need to
replace $\Psi_{\alpha}$ with a function $\tilde{\Psi}_{\alpha}$which is constant in a
small neighborhood of $\pm a$. In this case, when we write
$$\tilde{\Psi}_{\alpha}K(\kappa)\tilde{\Psi}_{\alpha}^{-1}=K(\kappa)+
\tilde{T}_\alpha$$
we see that $\tilde{T}_\alpha$ equals zero around the points where $V$ is
discontinuous. Therefore, if
$|\alpha|$ is small enough we will have
\begin{equation}\label{cdh40}
\Vert
K(\kappa)^k\tilde{\Psi}_{\alpha}R_\kappa(z)^{k}\tilde{\Psi}_{\alpha}^{-1}\Vert
<\infty,
\end{equation}
and the proof goes in the same way as in the previous section. 
\end{proof}

\section{Study of $\Xi_\eta$ and its adiabatic limit}

\setcounter{equation}{0}

Let us recall a few facts about the decoupled system: 
$$
\sigma_{sc}(\overset{\circ}{K}({\kappa}))=\emptyset,\qquad\mathcal{H}_{ac}(\overset{\circ}{K}({\kappa}))=
\mathcal{H}_-\oplus\mathcal{H}_+,\qquad\mathcal{H}_{pp}(\overset{\circ}{K}({\kappa}))=\mathcal{H}_0,
\qquad\forall\kappa\in[0,1],
$$
$$
\left.\overset{\circ}{H}\right|_{\mathcal{H}_{ac}(\overset{\circ}{H})}=
\Pi_-\big[\big(-\overset{\circ}{\Delta}_{D,-}\big)\big]\Pi_-\oplus\Pi_+\big
[\big(-\overset{\circ}{\Delta}_{D,+}\big)\big]\Pi_+
$$
$$
 \overset{\circ}{W}_{\eta}(s)E_{\rm
   ac}(\overset{\circ}{H})=\left[\Pi_-\left(e^{-iv_-\int_0^s\chi(\eta
       u)du}\right)+\Pi_+\left(e^{-iv_+\int_0^s\chi(\eta u)du}
   \right)\right]e^{-is\overset{\circ}{H}}E_{\rm ac}(\overset{\circ}{H})
$$
and using the notations defined earlier (\ref{sect-Lapl}), it is well known that: 
$$
\overset{\circ}{\Delta}_{D,\pm}=\mathfrak{l}_\pm\otimes1+1\otimes\mathfrak{L}_{\D},\quad\sigma(\mathfrak{l}_\pm)=
\sigma_{ac}(\mathfrak{l}_\pm)=[0,\infty),\qquad\sigma(\mathfrak{L}_{\D})=
\sigma_{pp}(\mathfrak{L}_{\D})\subset\mathbb{R}_+.
$$
Thus $\sigma_{ac}(\overset{\circ}{H})=[\inf
\sigma(\mathfrak{L}_{\D}),\infty)$ and has the set of thresholds 
$\mathcal{T}=\sigma_{pp}(\mathfrak{L}_{\D})$.

\subsection {Proof of Proposition \ref{P-2}}\label{ProofOfProposition2.8}%

Here we are interested in the strong limit when $s\to -\infty$ of:
$$
E_{\rm ac}(K_1)W_\eta^*(s)E_{\rm ac}(H)\overset{\circ}{W}_\eta(s)E_{\rm
  ac}(\overset{\circ}{H}).
$$
Let us start by noting that we can replace $E_{\rm ac}(H)$ with the
identity in the
above product, and still get the same strong limit (if it exists). The
explanation is that we can write:
$$E_{\rm pp}(H)\overset{\circ}{W}_\eta(s)E_{\rm
  ac}(\overset{\circ}{H})=\{E_{\rm pp}(H)e^{-is\overset{\circ}{H}} E_{\rm
  ac}(\overset{\circ}{H})\}\;e^{is\overset{\circ}{H}} \overset{\circ}{W}_\eta(s)$$
and use the fact that $e^{is\overset{\circ}{H}}
\overset{\circ}{W}_\eta(s)$ converges in norm, while $E_{\rm pp}(H)e^{-is\overset{\circ}{H}} E_{\rm
  ac}(\overset{\circ}{H})$ converges strongly to zero when $s\to -\infty$ because $E_{\rm pp}(H)$ is compact. 
Thus it is enough to study the existence of a strong limit when
$s\to -\infty$ of:
$$
\Xi_\eta(s):=E_{\rm ac}(K_1)W_\eta^*(s)\overset{\circ}{W}_\eta(s)E_{\rm
  ac}(\overset{\circ}{H}).
$$

For any $\delta>0$ let $\mathcal{V}_{\delta}$ be the set of vectors
$f\in\mathcal{H}_{ac}(\overset{\circ}{H})$ with compact spectral
support with respect to $\overset{\circ}{H}$ 
at distance larger than $\delta$ from all thresholds. Clearly, 
$\{\mathcal{V}_{\delta}\}_{\delta>0}$ is dense in 
$\mathcal{H}_{ac}(\overset{\circ}{K}_{\kappa})=\mathcal{H}_{ac}(\overset{\circ}{H})$. 
It is thus enough to show the existence of the limit
$\underset{s\searrow-\infty}{\lim}\Xi_{\eta}(s)f$ for
$f\in\mathcal{V}_{\delta}$. As any vector
$f\in\mathcal{H}_{ac}(\overset{\circ}{H})$ is of the form 
$(f_-,f_+)\in\mathcal{H}_-\oplus\mathcal{H}_+$ we will treat the two situations separately. 

The idea is to use a variant of Cook's method. We have the following identities:
\begin{align}\label{cdh36}
\Xi_{\eta}(s)f
&=E_{\rm ac}(K_1)W_\eta^*(s)\big(\overset{\circ}{K}(\chi(\eta s))+1\big)^{-2}
\overset{\circ}{W}_\eta(s)\big(\overset{\circ}{K}(\chi(\eta s))+1\big)^2f\nonumber \\
&=E_{\rm ac}(K_1)W_\eta^*(s)\big({K}(\chi(\eta s))+1\big)^{-1}
\big(\overset{\circ}{K}(\chi(\eta s))+1\big)^{-1}
\overset{\circ}{W}_\eta(s)\big(\overset{\circ}{K}(\chi(\eta s))+1\big)^2f\nonumber \\
&-E_{\rm ac}(K_1)W_\eta^*(s)\left[\big({K}(\chi(\eta s))+1\big)^{-1}-
\big(\overset{\circ}{K}(\chi(\eta s))+1\big)^{-1}\right]\overset{\circ}{W}_\eta(s)
\big(\overset{\circ}{K}(\chi(\eta s))+1\big)f\nonumber \\
&=:\Phi_\eta(s)-\Psi_\eta(s).
\end{align}
Without loss of generality, let us assume that $f\in
\mathcal{H}_+\cap \mathcal{V}_{\delta}$. Since $f$ is with compact
support in the spectral measure of $\overset{\circ}{H}$, there exist a
finite number $N$ of transverse eigenvectors $\{w_n\}$ of
$\mathfrak{L}_{\D}$ in $L^2(\D)$ corresponding to the eigenvalues
$\{\lambda_n\}$ and  so that 
\begin{equation}\label{cdh400}
f(x,{\bf x}_\perp)=\sum_{n=1}^N w_n({\bf x}_\perp)\int_{\mathbb{R}}\sin[k(x-a)]f_n(k)dk
\end{equation}
where $f_n$ are smooth, compactly supported, with a support which does
not contain the points $k^2<\delta$. Then we have
\begin{align}\label{cdh41}
\{e^{-is\overset{\circ}{H}} f\}(x,{\bf x}_\perp)&=\sum_{n=1}^N w_n({\bf
  x}_\perp)\int_{\mathbb{R}}
e^{-is(k^2+\lambda_n)}\sin[k(x-a)]f_n(k)dk,\nonumber \\
\{\overset{\circ}{W}_\eta(s) f\}(x,{\bf
  x}_\perp)&=e^{-iv_+\int_0^s\chi(\eta t)dt}\sum_{n=1}^N w_n({\bf
  x}_\perp)\int_{\mathbb{R}}
e^{-is(k^2+\lambda_n)}\sin[k(x-a)]f_n(k)dk.
\end{align}
Moreover, for $j\geq 1$ we have:
\begin{align}\label{cdh42}
&\{\overset{\circ}{W}_\eta(s)\big(\overset{\circ}{K}(\chi(\eta s))+1\big)^j f\}(x,{\bf
  x}_\perp)\\
&=e^{-iv_+\int_0^s\chi(\eta t)dt}\sum_{n=1}^N w_n({\bf
  x}_\perp)\int_{\mathbb{R}}
e^{-is(k^2+\lambda_n)}[k^2+v_+\chi(\eta
s)+\lambda_n]^j\sin[k(x-a)]f_n(k)dk.\nonumber
\end{align}
By standard integration by parts arguments, due to the support condition of $f_n$, we can prove the following
estimate:
 \begin{align}\label{cdh43}
&s^N \left \vert
  \{\overset{\circ}{W}_\eta(s)\big(\overset{\circ}{K}(\chi(\eta s))+1\big)^j f\}(x,{\bf
  x}_\perp)\right \vert\leq (1+|x-a|^N) C_N(f,j),\quad |s|>1,\,\forall N\in\mathbb{N} 
\end{align}
or (taking $N=2$)
\begin{align}\label{cdh44}
\left \Vert e^{-\alpha \langle Q_1\rangle }
  \{\overset{\circ}{W}_\eta(s)\big(\overset{\circ}{K}(\chi(\eta s))+1\big)^j
  f\}\right \Vert\leq \frac{C(f,j,\alpha)}{1+s^2}. 
\end{align}
We need only one more ingredient. Using the same methods as in subsection \ref{SSec-H0} 
one can prove the following estimation similar to \eqref{L2-descr-exp-0} for the given time
dependent objects:
\begin{equation}\label{cdh45}
\left\|e^{\alpha\langle Q_1\rangle }\Pi_\pm\big[ \big(K(\chi(\eta s))+1\big)^{-1} - \big(\overset{\circ}{K}(\chi(\eta s))+1\big)^{-1}\big]e^{\alpha\langle Q_1\rangle }\Pi_\pm\right\|\leq c,
\end{equation} 
for $\alpha$ small enough. The above constant can be chosen
uniformly with respect to $s$. 

Now we can go back to \eqref{cdh36} and investigate the structure of
$\Psi_\eta(s)$. Let us show that it will converge to zero when $s\to
-\infty$. Indeed, the difference of resolvents provides the
exponential localization near the sample. But then we know that the
adiabatic decoupled free evolution decays with $s$, as in
\eqref{cdh44}. We conclude:
 \begin{align}\label{cdh46}
\lim_{s\to-\infty}\Xi_{\eta}(s)f =\lim_{s\to-\infty}\Phi_\eta(s),
\end{align}
provided that the limit on the right hand side exists. We shall show that $\Phi_\eta(s)$ has an absolutely integrable
derivative with respect to $s$. Let us differentiate $\Phi_\eta(s)$ with respect to $s$. We obtain the identity:
\begin{align}\label{cdh47}
&-i\partial_s\Phi_\eta(s)=\\
&-E_{\rm ac}(K_1)W_\eta^*(s)
\left[\big({K}(\chi(\eta s))+1\big)^{-1}-\big(\overset{\circ}{K}
(\chi(\eta s))+1\big)^{-1}\right]
\overset{\circ}{W}_\eta(s)
\big(\overset{\circ}{K}(\chi(\eta s))+1\big)^2f\nonumber \\
&+i\eta\chi^\prime(\eta s)E_{\rm ac}(K_1)W_\eta^*(s)\big({K}(\chi(\eta
s))+1\big)^{-1}\nonumber \\
&\cdot V\left[\big({K}(\chi(\eta s))+1\big)^{-1}-\big(\overset{\circ}{K}
(\chi(\eta s))+1\big)^{-1}\right]
\overset{\circ}{W}_\eta(s)\big(\overset{\circ}{K}(\chi(\eta s))+1\big)f.\nonumber 
\end{align}
We see that using \eqref{cdh45} and \eqref{cdh46} we can write:
\begin{align}\label{cdh48}
||\partial_s\Phi_\eta(s)||\leq \frac{C}{1+s^2}. 
\end{align}
Thus $\lim_{s\to-\infty}\Phi_\eta(s)$ exists and equals:

\begin{align}\label{cdh49}
\Xi_\eta f&=\Phi_\eta(0)\\
&-i\int_{-\infty}^0E_{\rm ac}(K_1)W_\eta^*(s)
\left[\big({K}(\chi(\eta s))+1\big)^{-1}-\big(\overset{\circ}{K}
(\chi(\eta s))+1\big)^{-1}\right]
\overset{\circ}{W}_\eta(s)
\big(\overset{\circ}{K}(\chi(\eta s))+1\big)^2f ds\nonumber \\
&-\int_{-\infty}^0\eta\chi^\prime(\eta s)E_{\rm ac}(K_1)W_\eta^*(s)\big({K}(\chi(\eta
s))+1\big)^{-1}\nonumber \\
&\cdot V\left[\big({K}(\chi(\eta s))+1\big)^{-1}-\big(\overset{\circ}{K}
(\chi(\eta s))+1\big)^{-1}\right]
\overset{\circ}{W}_\eta(s)\big(\overset{\circ}{K}(\chi(\eta
s))+1\big)f ds.\nonumber 
\end{align}
The proof of
Proposition \ref{P-2} is over.

\subsection {Proof of Proposition \ref{P-4}}\label{ProofOfProp2.11}%

First, let us compute the limit $\eta\searrow 0$ in \eqref{cdh49}. We
can apply the Lebesgue dominated
convergence theorem in \eqref{cdh49} and obtain: 
\begin{align}\label{cdh51}
\lim_{\eta\searrow 0}\Xi_\eta f&=E_{\rm ac}(K_1)\big({K}(1)+1\big)^{-1}
\big(\overset{\circ}{K}(1)+1\big)f \\
&-i\int_{-\infty}^0E_{\rm ac}(K_1)e^{isK(1)}
\left[\big({K}(1)+1\big)^{-1}-\big(\overset{\circ}{K}
(1)+1\big)^{-1}\right]
e^{-is\overset{\circ}{K}(1)}
\big(\overset{\circ}{K}(1)+1\big)^2f ds.\nonumber 
\end{align}

Second, let us show that the above right hand side coincides with
$\Xi_0f$.  Indeed, let us look at the vector
$
E_{\rm ac}(K_1)e^{is K(1)}e^{-is\overset{\circ}{K}(1)}f,
$
where $f\in\mathcal{V}_\delta$. As before, we can decompose the vector
as :
$
\Phi_0(s)-\Psi_0(s)
$
where 
$$
\Phi_0(s):=E_{\rm ac}(K_1)e^{isK(1)}\big(K(1)+1\big)^{-1}\big(\overset{\circ}{K}(1)+1\big)^{-1}
e^{-is\overset{\circ}{K}(1)}\big(\overset{\circ}{K}(1)+1\big)^2f,
$$
$$
\Psi_0(s):=E_{\rm ac}(K_1)e^{isK(1)}\left[\big(K(1)+1\big)^{-1}-\big(\overset{\circ}{K}(1)+1\big)^{-1}
\right]e^{-is\overset{\circ}{K}(1)}\big(\overset{\circ}{K}(1)+1\big)f.
$$

Using the previous propagation estimates which were shown to be
uniform in $\eta$, we can repeat the same argument which led us to
\eqref{cdh49} but with $\eta=0$ from the beginning. This will give a
formula for $\Xi_0f$ which will coincide with the right hand side of
\eqref{cdh51}. The proof is over.

\section{Acknowledgments}

Part of this
work has been done while P.D. and R.P. were visiting professors at Aalborg
University. H.C. acknowledges support from the Danish FNU grant {\it Mathematical Physics}. R.P. aknowledges the CNCSIS support
under the Ideas Programme, PCCE project no. 55/2008 {\it “Sisteme
diferentiale in analiza neliniara si aplicatii”}.

\end{document}